%% file: paper.tex
\documentclass[submission,copyright,creativecommons]{eptcs}
\providecommand{\event}{GandALF 2022} 

\usepackage{amsmath,amsmath,amsfonts,amssymb,amsthm}
\usepackage{verbatim}

\usepackage{thmtools}
\usepackage{thm-restate}

\input{mathalg.tex}

\newif\ifabstract
\abstractfalse
\newif\iffull
\ifabstract \fullfalse \else \fulltrue \fi

\ifabstract

\else

\fi

\newcounter{section-preserve}
\newcounter{lemma-preserve}
\newcounter{corollary-preserve}
\newcounter{theorem-preserve}
\newcommand{\blank}[1]{}
\newtoks\magicAppendix
\magicAppendix={}
\newtoks\magictoks
\newif\iflater
\laterfalse
\long\def\later#1{\iflater#1\else\magictoks={#1}%
	\edef\magictodo{\noexpand\magicAppendix={\the\magicAppendix \par
			\the\magictoks%
	}}
	\magictodo\fi}
\long\def\both#1{\iflater#1\else\magictoks={#1}%
	\edef\magictodo{\noexpand\magicAppendix={\the\magicAppendix \par
			\noexpand\setcounter{theorem-preserve}{\noexpand\arabic{theorem}}%
			\noexpand\setcounter{theorem}{\arabic{theorem}}%
			\noexpand\setcounter{lemma-preserve}{\noexpand\arabic{lemma}}%
			\noexpand\setcounter{lemma}{\arabic{lemma}}%
			\noexpand\setcounter{corollary-preserve}{\noexpand\arabic{corollary}}%
			\noexpand\setcounter{corollary}{\arabic{corollary}}%
			\noexpand\setcounter{section-preserve}{\noexpand\arabic{section}}%
			\noexpand\setcounter{section}{\arabic{section}}%
			\noexpand\let\noexpand\oldsection=\noexpand\thesection
			\noexpand\def\noexpand\thesection{\thesection}
			\noexpand\let\noexpand\oldlabel=\noexpand\label
			\noexpand\let\noexpand\label=\noexpand\blank
			\the\magictoks%
			\noexpand\setcounter{theorem}{\noexpand\arabic{theorem-preserve}}%
			\noexpand\setcounter{lemma}{\noexpand\arabic{lemma-preserve}}%
			\noexpand\setcounter{corollary}{\noexpand\arabic{corollary-preserve}}%
			\noexpand\setcounter{section}{\noexpand\arabic{section-preserve}}%
			\noexpand\let\noexpand\thesection=\noexpand\oldsection
			\noexpand\let\noexpand\label=\noexpand\oldlabel
	}}
	\magictodo
	\the\magictoks\fi}
\def\magicappendix{} 

\begin{document}

\title{Characterizing the Decidability of Finite State Automata Team Games with Communication}
\def\titlerunning{Decidability of Finite Automata Team Games with Communication}

\author{Michael Coulombe
\institute{Department of Electrical Engineering and Computer Science, Massachusetts Institute of Technology, Cambridge, MA, USA}
\email{mcoulomb@mit.edu}
\and Jayson Lynch
\institute{Cheriton School of Computer Science, University of Waterloo, Waterloo, Ontario, Canada}
\email{jayson.lynch@uwaterloo.ca}
}
\def\authorrunning{M. Coulombe, J. Lynch}


\maketitle

\begin{abstract}
In this paper we define a new model of limited communication for multiplayer team games of imperfect information. We prove that the Team DFA Game and Team Formula Game, which have bounded state, remain undecidable when players have a rate of communication which is less than the rate at which they make moves in the game. We also show that meeting this communication threshold causes these games to be decidable.
\end{abstract}

\section{Introduction}

Deciding optimal play in multiplayer games of incomplete information is known to be an undecidable problem~\cite{peterson_multiple-person_1979,peterson2000multiplayer}. This includes games where the state space is bounded, a surprising result first shown of a collection of abstract computation games~\cite{demaine2008constraint} that has been extended to generalized versions of real games, like Team Fortress 2 and Super Smash Bros~\cite{coulombe_et_al:LIPIcs:2018:8805}. However, past work has relied on the complete inability of teammates to communicate during the game, which is often not a realistic assumption.
In this paper we study deterministic models of communication between players in two of these computation games, the Team DFA Game and the Team Formula Game, and show a sharp change in computational complexity based on whether players are able to eventually communicate all of their moves or only able to communicate a constant fraction.

One motivation for this model is a better understanding of real world games.
Many team games played in-person naturally permit free form communication between teammates to coordinate their actions, and
online multiplayer video games often provide communication channels such as voice-chat, text, and emotes to simulate this in-person environment. These include many FPS games such as Team Fortress and Left4Dead, MOBAs such as DOTA2 and League of Legends, and RTS games such as Starcraft and Age of Empires. Some of these examples have drawn research interest in AI/ML~\cite{alphastarblog,berner2019dota} as well as computational complexity\cite{viglietta2014gaming,coulombe_et_al:LIPIcs:2018:8805}. The real-time nature of these games ensures that communication channels are bounded; 
however, modeling free form communication, as well as efficiently implementing meaningful player choices in a real-time setting, makes it difficult to analyze these games with these communication features enabled. 

Outside of the team setting, communication is a central aspect of many other games. For example, in the cooperative card game Hanabi players are unable to see their own hands of cards, but this information is visible to everyone else. In addition, players are not allowed to communicate except through actions in the game, one of which allows players to reveal partial information about what is in another players hand. A perfect information version of Hanabi was shown to be NP-complete~\cite{baffier2016hanabi}. The Crew: Quest for Planet Nine is a cooperative trick taking card game which uses limited communication between players as a core mechanic. The complexity of this game was also studied in the perfect information setting~\cite{reiber2021crew}. Under the limited information setting, containment in NP for both of these games is not obvious, and we see a need for models of player communication in games. Other examples of cooperative boardgames with limited communication channels between players include Mysterium, The Mind, and Magic Maze.

Other games may limit communication simply with time pressure in the game. Examples of fully cooperative games with imperfect information that use time pressure to limit coordination and communication include Space Alert, 5-Minute Dungeon, Keep Talking and Nobody Explodes, and Spaceteam.

Multi-agent, imperfect information games are also a topic of interest in Reinforcement Learning. In \cite{celli2019coordination} algorithms are developed to address team extensive form games of imperfect information where communication is allowed at certain points during gameplay, with Bridge and collusion among some players between hands in poker being the motivating examples. Other work considers Sequential Social Dilemmas, a type of iterated economic game where in any given instant a player is incentivized towards non-cooperative behavior, but cooperative strategies can obtain higher payoff over the game as a whole. Learning algorithms for these models both with and without explicit bounded communication channels were studied in \cite{pmlr-v97-jaques19a}. Purely cooperative settings have also been of interest~\cite{foerster2016learning}.

One major achievement was human level performance on a limited version of DOTA2, a MOBA-style video game~\cite{berner2019dota}. These are real-time, team games with partially observable state. Although both text and voice chat are typically allowed in professional play, the AI system did not utilize these explicit communication mechanisms. The board game Diplomacy, while not explicitly a team game, features coordination and temporary alliances between players as a core game dynamic. This game has also seen interest as a new challenge in the AI community, but focusing on No-Press Diplomacy which does not allow explicit communication between players~\cite{paquette2019no,bakhtin2021no}.

These examples of both AI and computational complexity research which considers games with cooperation and communication motivate, but frequently ignore the important role of communication in these games, motivates this paper.

\paragraph*{Paper Organization.} In Section~\ref{sec: Communication Model} we formally define our model of communication for the Team DFA Game. In Section~\ref{sec:basic} we prove undecidability for a few simple communication patterns to help build intuition for the techniques used in the next section. In Section~\ref{sec:general} we prove our main undecidability result for Team DFA Games with Communication. In Section~\ref{sec: Decidability} we prove the game becomes decidable when either both players can communicate all information about their moves, or one player receives no information but can communicate all of their moves to their teammate. In Section~\ref{sec: Team Formula Game} we show that analogous algorithms and undecidability results hold for Team Formula Games.

\subsection{Team DFA Game}

The Team DFA Game is a bounded state, two-vs-one multiplayer game, defined by \cite{coulombe_et_al:LIPIcs:2018:8805} as a simplification of the undecidable Team Computation Game \cite{demaine2008constraint}. It involves a team of two players, the $\exists$~players, and an adversary, the $\forall$~player, who take turns sending a bit to a deterministic finite automaton (DFA).
Each team's goal is to put the DFA into any of a set of winning final states for their team.
The $\forall$~player knows the moves of the $\exists$~players and thus the state of the DFA; however the $\exists$~players do not know each other's moves, which they must make after learning only one of $\forall$'s moves each turn.

\begin{define}
	The Team DFA Game (TDG) is a two-versus-one team game.
	An instance of the game is a DFA $D = (\{0,1\}, Q, q_0, \delta, F = F_\exists \cup F_\forall)$.
	The existential team $\{\exists_0,\exists_1\}$ competes against the universal team $\{\forall\}$.
	The game starts with $D$ in state $q_0$ and each round proceeds as follows:
	\begin{enumerate}
		\item If D's state $q \in F_\exists$ then team existential wins. If $q \in F_\forall$ then team universal wins.
		\item $\forall$ learns the state $q$ of $D$ then inputs two bits $b_0,b_1$ into $D$.
		\item $\exists_0$ learns $b_0$ then inputs one bit $m_0$ into $D$. $\forall$ learns $m_0$.
		\item $\exists_1$ learns $b_1$ then inputs one bit $m_1$ into $D$. $\forall$ learns $m_1$.
	\end{enumerate}
\end{define}

The problem we consider in this paper is: given an instance of the game, does a particular team have a \emph{forced win}?
More formally, does there exist a strategy function $s_i$ for each player $i$ on the team, specifying on each turn which move to make based on any information they have learned so far, that when followed will guarantee that this team will win?
In this paper, we define the complexity of a game as the complexity of whether a specified team has a forced win in the game, such as in the following:

\begin{theorem}[previous work \cite{coulombe_et_al:LIPIcs:2018:8805}]
	The Team DFA Game is undecidable.
\end{theorem}

A number of variations of this game, all undecidable in the general case, exist. These include Team Computation Game where players give inputs to a Turing machine on a bounded tape~\cite{peterson_multiple-person_1979}, Team Constraint Logic Game where players make moves in a partially observable constraint logic graph~\cite{demaine2008constraint}; and Team Formula Game where players flip the values of Boolean variables trying to satisfy different formulas~\cite{peterson2000multiplayer,demaine2008constraint}.

\subsection{Communication Model}
\label{sec: Communication Model}

\begin{figure}[bp]
\centering
\includegraphics[width=.75\textwidth]{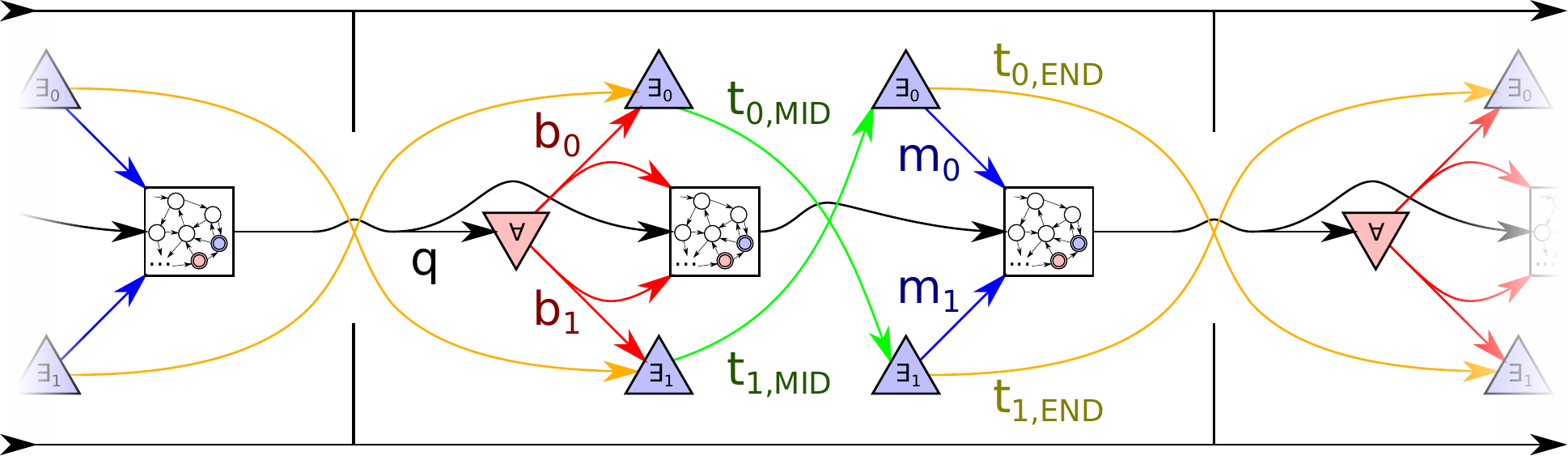}
\caption{Information flow graph of one round of the Team DFA Game with Communication, including from the previous round and into the next round. 
New to this game are the mid-round transmissions, $t_{0,\textsc{mid}}$ and $t_{1,\textsc{mid}}$, and the end-of-round transmissions, $t_{0,\textsc{end}}$ and $t_{1,\textsc{end}}$, which have sizes determined by $P_\textsc{mid}$ and $P_\textsc{end}$ applied to the policy state.
}
\label{fig:tdgc_dag}
\end{figure}

We model communication in the Team DFA Game with a policy that specifies the bandwidth of a dynamic information channel, as one might have due to natural factors (e.g. playing a real-time game with voice chat) or intentional game design (e.g. Hanabi's card-revealing moves) allowing a predictable but bounded amount of player-to-player communication between moves. Specifically, a policy $P$ is a DFA over a unary alphabet with functions $P_\textsc{mid}, P_\textsc{end}$ over its states. In a round of the game in policy state $p$, 
$P_\textsc{mid}(p)$ is the number of bits which are exchanged simultaneously between $\exists_0$ and $\exists_1$ after $(b_0,b_1)$ are revealed but before $(m_0,m_1)$ must be determined,
and similarly $P_\textsc{end}(p)$ is the number of bits to exchange after $(m_0,m_1)$ are submitted but before the next round starts.

\begin{define}
    The Team DFA Game with Communication (TDGC) is a game of
    the existential team $\{\exists_0,\exists_1\}$ versus the universal team $\{\forall\}$, extending the Team DFA Game.
    An instance of the game is a pair of a game DFA $D = (\{0,1\}, Q, q_0, \delta, F_\exists \cup F_\forall)$ and a policy $P$, which consists of a policy DFA $(\{1\},\Pi,p_0,\pi,\emptyset)$ and functions $P_\textsc{mid}, P_\textsc{end} : \Pi \to \mathbb{N} \times \mathbb{N}$.
    The game starts with $D$ in state $q_0$ and the policy DFA in state $p_0$, and each round proceeds
    with added communication steps
    as illustrated in Figure~\ref{fig:tdgc_dag}.

\end{define}

\later{ 
	\begin{algorithm}
		\caption{Team DFA Game: execution of one round, given $D$ is in state $q$ and $P$ is in state $p$.}
		\label{alg:tdg-def}
		\begin{algorithmic}[1]
			\Function{team-dfa-game-round}{$q$, $p$}
			\State If $q \in F_\exists$, then the existential team wins.
			
			\State If $q \in F_\forall$, then the universal team wins.
			
			\State $\forall$ learns $q$, then inputs two bits $b_0,b_1$ into $D$. 
			\Comment{$q \gets \delta(\delta(q,b_0),b_1)$}
			
			\State $\exists_0$ learns $b_0$, and $\exists_1$ learns $b_1$.
			
			\State 
			\Call{exchange}{$P_\textsc{mid}(p)$}
			
			\State $\exists_0$ inputs one bit $m_0$ into $D$. 
			\Comment{$q \gets \delta(q,m_0)$}
			
			\State $\exists_1$ inputs one bit $m_1$ into $D$.
			\Comment{$q \gets \delta(q,m_1)$}
			
			\State $\forall$ learns $m_0$ and $m_1$.
			
			\State 
			\Call{exchange}{$P_\textsc{end}(p)$}
			
			\State Advance the policy state. \Comment{$p \gets \pi(p,1)$}
			\EndFunction
		\end{algorithmic}

		\begin{algorithmic}[1]
			\Function{exchange}{$(n_{01}, n_{10})$}
			\State $\exists_0$ privately defines message $t_0$, consisting of $n_{01}$ bits.
			\State $\exists_1$ privately defines message $t_1$, consisting of $n_{10}$ bits.
			\State $\exists_0$ and $\forall$ learn $t_1$.
			\State $\exists_1$ and $\forall$ learn $t_0$.
			\EndFunction
		\end{algorithmic}
\end{algorithm}
}

There are two beneficial aspects of studying policies as DFAs on unary alphabets: bounding the state space allows for the policy to be computable by the mechanics of a bounded-state game (such as the DFA in the Team DFA Game), and giving every state exactly one next state (for the next round of the game) means the bandwidth every round will be known in advance when building our constructions, rather than being dependent on player actions.
As a result of this choice, it is important to note that the shape of its state transition graph will always have the form: from the start state, there is an initial chain of unique states (possibly of length zero) that leads to a cycle of periodically-repeating states. Also shown in Figure~\ref{fig:policy-dfa}.

\begin{figure}[ht]
\centering
\includegraphics[scale=1.00]{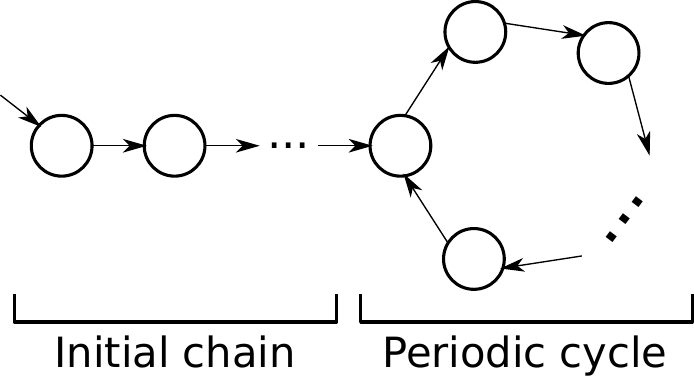}
\caption{General form of a policy DFA: an initial chain followed by a cycle.}
\label{fig:policy-dfa}
\end{figure}

\section{Undecidability of Simple Communication Games}
\label{sec:basic}

\ifabstract
\later{
\section{Undecidability of Simple Communication Games}
\label{app:basic}}
\fi

In this section, we will explore some basic classes of policies that preserve the undecidability of the Team DFA Game with Communication. Our proof technique is to reduce from the zero-communication Team DFA Game, where we compensate for the message passing by ``clogging the channel" with the forced transmission of garbage bits that do not facilitate information sharing. Section~\ref{sec:general} builds upon these examples to obtain more general results, however proving the special cases in this section allows us to introduce ideas needed in the full proof and discuss some of the techniques more concretely.

For each class of policies $\mathcal{P}$ 
below, we will show that given any policy $P \in \mathcal{P}$ and DFA $D$ for playing TDG, 
we can produce a DFA $D'$ for playing TDGC under $P$ such that the $\exists$ team has a forced win on $D$ with no channel iff they have a forced win on $D'$ given a channel following policy $P$. As TDG is undecidable, so will be TDGC under any policy $P \in \mathcal{P}$.
For simplicity, we consider policies with DFA $C_r$, a length-$r$ cycle of states $\Pi = \{0,1,\ldots,r-1\}$ with no initial chain, for arbitrary values of $r$.

\iffull
\fi

\begin{figure}[htp]
\centering
\includegraphics[scale=0.80]{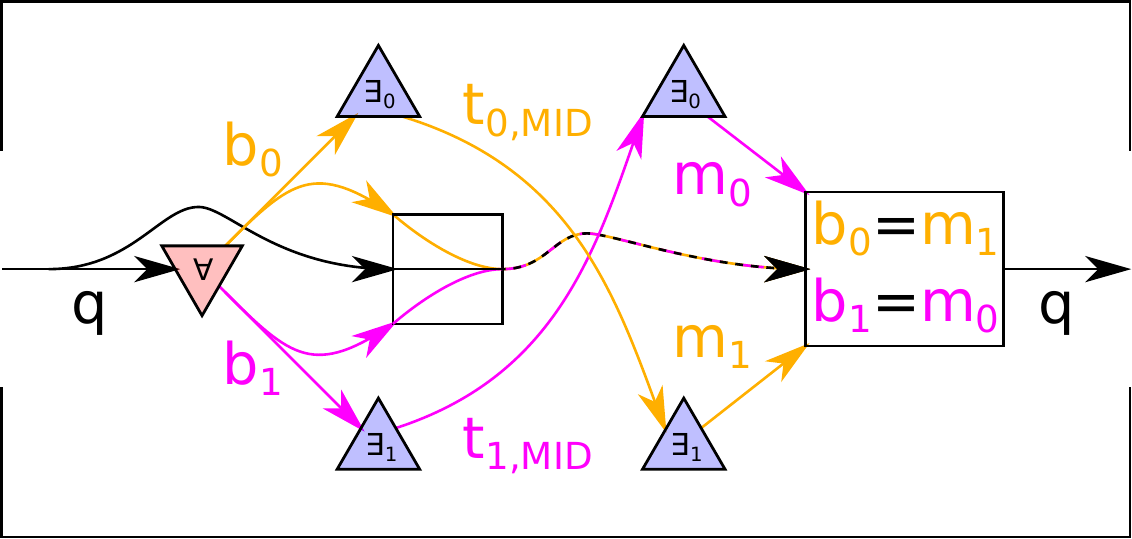}
\caption{Mid-round 1-bit channel clogging technique. Values with the same color must be equal, namely $t_i = b_i = m_{1-i}$, or else the DFA permanently enters $F_\forall$.}
\label{fig:mid_clog}
\end{figure}

\begin{theorem}\label{thm:mid-1-per-r}
TDGC is undecidable with a 1 bit mid-round exchange every $r \geq 2$ rounds:
policies $P$ where
$P_\textsc{mid}(p) = (1,1)$ if $p \equiv 0 \mod r$, $P_\textsc{mid}(p) = (0,0)$ otherwise,
and
$P_\textsc{end}(p) = (0,0)$.
\end{theorem}

\begin{proof}
	We construct a DFA $D'$ by first augmenting the state $q$ of $D$ with the state $p$ of $C_r$. When $p \not\equiv 0 \mod r$, $D'$ simply simulates $D$ for one round. However, when $p \equiv 0 \mod r$, $D'$ instead takes the inputs $(b_0, b_1, m_0, m_1)$ and tests $b_0 = m_1 \land b_1 = m_0$. If the test fails, then $D'$ enters a final state for $\forall$.
	
	How $D'$ clogs the channel is diagrammed in Figure~\ref{fig:mid_clog}. By tracking the round number, it knows exactly when $\exists_0$ and $\exists_1$ will exchange bits, and in that round $D'$ expects $\exists_0$ to guess $b_1$, a bit that $\exists_0$ does not learn by the game procedure, and vice-versa. $\exists_0$ and $\exists_1$ are forced to spend their single bit of communication on exchanging $b_0$ and $b_1$ to their teammate, in order to guarantee survival against any $\forall$ strategy for choosing $b_0$ and $b_1$.
	
	Since $\exists_0$ and $\exists_1$ do not learn anything from each other or alter the simulated $D$'s state in the rounds with communication, they have a winning strategy on $D'$ playing TDGC under $P$ if and only if they have a strategy for the non-exchanging rounds (which happen infinitely-often since $r \geq 2$) that would give a winning strategy on $D$ playing TDG.
\end{proof}


\begin{theorem}\label{thm:mid-n-ones-per-r}
TDGC is undecidable with $n$ rounds of 1-bit mid-round exchanges across $r>n$ rounds:
policies $P$ where 
$P_\textsc{mid}(p) \in \{(0,0),(1,1)\}$
with
pre-image size $|P_\textsc{mid}^{-1}((1,1))| = n$
and
$P_\textsc{end}(p) = (0,0)$.
\end{theorem}

\begin{proof}
	We generalize Theorem~\ref{thm:mid-1-per-r} by constructing a DFA $D'$ which clogs the channel on any round $p \pmod{r}$ in which $P_\textsc{mid}(p) = (1,1)$ and simulates $D$ in the other $r-n>0$ out of $r$ rounds.
	By the same argument, this prevents communication between $\exists_0$ and $\exists_1$ while playing TDGC beyond the corresponding play of TDG taking place during non-exchanging rounds, and thus preserves forced win-ability.\looseness=-1
\end{proof}

\later{
\iffull
\subsection{End-round Communication}
\fi

\begin{figure}[htp]
\centering
\includegraphics[width=\textwidth]{figures/dfagame/end_order_clog_r_geq_3.pdf}
\caption{End-round channel clogging technique when $r \geq 3$, showing two rounds. The faded-out edges represent messages $(m_0,m_1,b_0',b_1')$, which are not used. $D'$ simulates $D$ on other rounds.}
\label{fig:end_clog_geq_3}
\end{figure}

\begin{figure}[htp]
\centering
\includegraphics[width=\textwidth]{figures/dfagame/end_order_clog_r_eq_2.pdf}
\caption{End-round channel clogging technique when $r=2$, showing three rounds. Bits $(b_0',b_1')$ created in odd rounds get checked two rounds later, labeled as $(b_0^*,b_1^*)$. $D'$ simulates $D$ in even rounds before $\exists$~players get a chance to exchange those bits.
 }
\label{fig:end_clog_eq_2}
\end{figure}

\begin{theorem}
\label{thm:end_clog_geq_3}
TDGC is undecidable with a 1 bit end-round exchange every $r \geq 3$ rounds: 
policies $P$ where 
$P_\textsc{mid}(p) = (0,0)$
and
$P_\textsc{end}(p) = \begin{cases}
	(1,1), &\text{if } p \equiv 0 \mod r\\
	(0,0) & \text{otherwise}\\
\end{cases}$.
\end{theorem}
\begin{proof}
	We construct a DFA $D'$ like before, by augmenting $D$ with $C_r$'s state as well as two bits of storage, initialized to $(0,0)$.
	When $p \equiv 0 \mod r$, $D'$ stores the inputs $(b_0,b_1)$ from $\forall$ and ignores $\exists_0$ and $\exists_1$ in that round. In the next round, when $p \equiv 1 \mod r$, it will ignore $\forall$ in that round and use those stored bits to test $b_0 = m_1 \land b_1 = m_0$. If the test fails, then $D'$ enters a final state for $\forall$. 
	In all other rounds $\{2,\ldots,r-1\} \mod r$, $D'$ simply simulates $D$.
	
	The exchange occurring between round $0$ and $1$ (mod $r$) requires this new clogging method, shown in Figure~\ref{fig:end_clog_geq_3}. Compared to the constructions in Theorems~\ref{thm:mid-1-per-r} and \ref{thm:mid-n-ones-per-r}, we separate the bit choices of $\forall$ and the bit guesses of $\exists_0$ and $\exists_1$ into two rounds to leave time for communication of the bits, without using up all messages that are otherwise used to simulate $D$ (since $r \geq 3$).

	By similar arguments to previous theorems, $D'$ clogs all significant channels thus this reduction preserves forced win-ability.
\end{proof}


\begin{theorem}
\label{thm:end_clog_eq_2}
TDGC is undecidable with a 1 bit end-round exchange every $r=2$ rounds: policies $P$ where
$P_\textsc{mid}(p) = (0,0)$
and
$P_\textsc{end}(p) = \begin{cases}
	(1,1), &\text{if } p \equiv 0 \mod 2\\
	(0,0) &\text{if } p \equiv 1 \mod 2\\
\end{cases}$.
\end{theorem}
\begin{proof}
    The construction in Theorem~\ref{thm:end_clog_geq_3} for $r \geq 3$ fails when $r=2$ because it requires two rounds without $D'$ simulating $D$.
    To address this, we give the modified construction shown in Figure~\ref{fig:end_clog_eq_2}.
    We augment $D$'s state with two pairs of bits, to remember the two most recent transmissions.
    In an odd round, $D'$ stores $(b_0,b_1)$ from $\forall$ as the most recent pair, and then expects the $\exists$~team to submit bits equal to the least recent pair.
    In an even round, $D'$ simulates $D$.
    To deal with the first two rounds, where there is no previous transmission to validate, we initializing all stored bits to $0$ so the $\exists$~players can just submit $0$ to pass the check.
    
    Although this technique does have in-flight clogging bits during the simulation rounds, every transmission occurs directly before a validation check for those bits, which guarantees that the $\exists$~players must exchange the clogging bits rather than attempt to communicate other information to gain an advantage in the TDG on $D$.
    Therefore, as before, this reduction implies undecidability for $r=2$ as well.
\end{proof}

\begin{figure}[htp]
\centering
\includegraphics[scale=1.00]{figures/tdg_r-k_eq_1.pdf}
\caption{Pattern for the end-round channel clogging gadget. Matching colors denote the flow of clogging bits. The first round's $(m_0,m_1)$ and last round's $(b_0,b_1)$ are unused, and the last round may or may not include an end-exchange.}
\label{fig:tdg_r-k_eq_1}
\end{figure}

\begin{theorem}
TDGC is undecidable with $n$ bits worth of only end-round exchanges across $r>n$ rounds:
where
$P_\textsc{mid}(p) = (0,0)$,
$P_\textsc{end}(p) \in \{(0,0),(1,1)\}$,
and
pre-image size $|P_\textsc{end}^{-1}((1,1))| = n$.

\end{theorem}
\begin{proof}
	Again reducing from TDG, since $r>n$, we have at least one round per period with no end-exchange, and assuming $n \geq 1$ at least one round with an end-exchange (otherwise this is equivalent to TDG).
	We combine the techniques from Theorems~\ref{thm:end_clog_geq_3}~and~\ref{thm:end_clog_eq_2} to construct a $D'$ which handles each stretch of contiguous rounds with end-round exchanges and those without.
	
	Our gadget which does this is described by the pattern in Figure~\ref{fig:tdg_r-k_eq_1}.
	Leading up to a round with no end-exchange, we repeat the technique from Figure~\ref{fig:end_clog_geq_3} (possibly zero times), using the previously ``unused'' bits to overlap.
	In the last round in the stretch of no-end-exchange rounds (which may be the only such round), the pattern finishes with the technique from Figure~\ref{fig:end_clog_eq_2}, guaranteeing one round is available to simulate $D$ per period.
	Following this, another instance of the gadget begins, repeating ad infinitum.
	
	The clogging guarantees from Theorems~\ref{thm:end_clog_geq_3}~and~\ref{thm:end_clog_eq_2} still hold of this combined gadget, since the validation of each clogging bit still occurs directly after it can be first exchanged. Inductively assuming the previous checks successfully clogged earlier transmissions, the $\exists$~players cannot communicate any strategically beneficial information to each other, so whether or not they have a winning strategy is undecidable.
\end{proof}

}

\section{Undecidability of General Communication Games}
\label{sec:general}
\ifabstract
\later{
\section{Undecidability of General Communication Games}
\label{app:general}
}
\fi

This section proves our main result: that a broad class of policies with sufficiently low communication rate remain undecidable for the Team DFA game. We now define this more general notion.

\begin{define}
	A policy is \textbf{$(r,x_0,x_1,N)$-rate-limited} if, after a fixed number of rounds $N$, the rate of transmission from player $\exists_i$ to $\exists_{1-i}$ is $x_i$ during every period of $r$ rounds. Specifically, it must satisfy
	$x_i = \sum\limits_{k = k_0}^{k_0 + r-1} P_\textsc{mid}(p_k)[i] + P_\textsc{end}(p_k)[i]$
	for $k_0 = N + \ell r$,
	where
	$\ell \in \mathbb{N}$
	and
	$p_k$ is the policy state on round $k$.
\end{define}

This now allows us to state the main theorem.

\begin{restatable}[]{theorem}{undecidable}
	\label{thm:tdgc-undecidable-xi-lt-r}
	TDGC is undecidable under all $(r,x_0,x_1,N)$-rate-limited policies where $x_0,x_1 < r$.
\end{restatable}


\subsection{Properties of Rate-Limited Policies}
\label{sec:Properties of Rate-Limited Policies}

\ifabstract
\later{
\subsection{Proofs of Rate-Limited Policies}
\label{app:Proofs of Rate-Limited Policies}
}
\fi

Before proceeding to the proof of Theorem~\ref{thm:tdgc-undecidable-xi-lt-r}, we will establish useful lemmas about rate-limited policies. \ifabstract Proofs can be found in Appendix~\ref{app:Proofs of Rate-Limited Policies}.\fi

First, we have the following two simple observations:

\both{
\begin{lemma}
\label{lem:policy-dfa-rate-limited}
	Any policy implemented as a unary-alphabet DFA with $n>1$ states is $(r,x_0,x_1,N)$-rate-limited for some $1 < r \leq n$ and some initial segment of length $N \leq n$.
\end{lemma}
}

\later{
\begin{proof}
	Refer back to the diagram in Figure~\ref{fig:policy-dfa}.
	Consider the unary-alphabet DFA as a directed graph on states, each with exactly one outgoing edge representing the next transition. Any path starting at the start state $p_0$ cannot have length at least $n$ without repeating some state $p_N$, so all sufficiently-long paths must first traverse an initial segment $(p_0,\ldots,p_{N-1})$ of $N \leq n$ states then go around the simple cycle $(p_N,\ldots,p_{N+r},p_N)$ of $r \leq n$ states forever after. If the simple cycle is a self-loop, then instead use the non-simple cycle composed of going around the self-loop twice to have $r>1$.
	
	Going around this cycle of $r$ states once will permit transmission of a fixed number of bits $x_i = \sum\limits_{k = N}^{N+r} P_\textsc{mid}(p_k)[i] + P_\textsc{end}(p_k)[i]$ from player $\exists_i$ to $\exists_{1-i}$, thus defining a period that makes the policy $(r,x_0,x_1,N)$-rate-limited.
\end{proof}
}

\both{
\begin{lemma}
\label{lem:policy-period-double}
	Any $(r,x_0,x_1,N)$-rate-limited policy is also $(2r,2x_0,2x_1,N)$-rate-limited if $r>1$.
\end{lemma}
}

\later{
\begin{proof}
	Given a cycle of length $r$ starting after $N$ steps, repeating the cycle twice results in length $2r$ periods with $2x_0$ and $2x_1$ bits of transmissions also starting after $N$ steps.
\end{proof}
}

Next, we will need the following property bounding the partial sums of certain repeated finite sequences for analyzing the transmission rates in each part of a round across a period.

\begin{define}
	Let $a$ be any sequence of $2n$ natural numbers
	$(a_0,a_1,\ldots,a_{2n-1})$
	with sum at most $n-1$, and let $i$ be an index into $a$.
	\Call{rotate-bounded}{$a, i$} is the predicate that holds when the infinite sequence
	$b^{(i)}_j = a_{(i+j \mod 2n)}$
	with partial sums
	$B^{(i)}_j = \sum\limits_{k=0}^{j-1} b^{(i)}_k$
	satisfies
	$\forall j>0.\ B^{(i)}_j < \frac{j}{2}$.
\end{define}

\begin{lemma}
\label{lem:partials-leq-linear-even}
	For any such $a$, \Call{rotate-bounded}{$a, i$} holds for some even index $i$.
\end{lemma}

\begin{proof}
    Let $C^{(i)}_j = B^{(i)}_j - \frac{j}{2}$. Let us find an even index $i$ such that $C^{(i)}_j < 0$ for all $j>0$,
    and consider the largest length $j^*<2n$ which maximizes $C^{(0)}_{j^*}$.
    If $C^{(0)}_{j^*} < 0$ then $i=0$ satisfies the claim,
    so suppose $C^{(0)}_{j^*} \geq 0$.
    
    Because the sum of $a$ is at most $n-1$, notice that for any $j$, $j+2n$ in the next period satisfies
    $C^{(0)}_{j+2n}
    = B^{(i)}_{j+2n} - \frac{j+2n}{2}
    < \left(B^{(i)}_{j} + n\right) - \left(\frac{j}{2} + n\right)
    = C^{(0)}_{j}$.
    Since $j^*$ is the largest maximizer in the first period,
    for all $j > j^*$:
    \begin{equation*}
	    0
	    > C^{(0)}_{j} - C^{(0)}_{j^*}
	    = \left(B^{(0)}_j - \frac{j}{2}\right) - \left(B^{(0)}_{j^*} - \frac{j^*}{2}\right)
	    = B^{(j^*)}_{j-j^*} - \frac{j - j^*}{2}
	    = C^{(j^*)}_{j-j^*}\\
    \end{equation*}
    and therefore $\forall j > 0.\ C^{(j^*)}_{j} < 0$,
    thus $i=j^*$ is an index for which $a$ is rotate bounded. If $j^*$ is even we are done.
    If $j^*$ is odd, then we know that $j^*<2n-2$ because we supposed that $C^{(0)}_{j^*} \geq 0$ whereas $C^{(0)}_{2n-1} < 0$:
    \begin{equation*}
    C^{(0)}_{2n-1}
    = B^{(0)}_{2n-1} - \frac{2n-1}{2}
    \leq \left(n-1 - a_{2n-1}\right) - \left(n - \frac{1}{2}\right)
    = - a_{2n-1} - \frac{1}{2}
    < 0
    \end{equation*}
    Thus, consider the even $j^*+1$ and let $j'\in [j^*+1,2n)$ be the maximum length such that $C^{(0)}_{j'} = C^{(0)}_{j^*+1}$. 
    $
	    0
	    > C^{(0)}_{j^*+1} - C^{(0)}_{j^*}
	    = C^{(j^*)}_{1}
	    = b^{(j^*)}_{0} - \frac{1}{2}$
	    therefore
	    $b^{(j^*)}_{0} = 0$
	    and
	    $C^{(0)}_{j'}
	    = C^{(0)}_{j^*+1}
	    = C^{(0)}_{j^*} - \frac{1}{2}
    $. 
    Since $B^{(0)}_{j^*}$ is an integer and $j^*$ is odd, $j'$ must also be even,
    and by similar arguments $\forall j > j'.\ C^{(0)}_{j'} > C^{(0)}_{j}$,
    thus $\forall j > 0.\ C^{(j')}_{j} < 0$,
    so $i=j'$ satisfies the claim.
\end{proof}

\later{
\begin{figure}[htp]
\centering
\includegraphics[width=0.45\textwidth,clip,trim=0.4in 0.6in 0.4in 0.8in]{figures/plot_rot_seq_even.pdf}
\hspace{1em}
\includegraphics[width=0.45\textwidth,clip,trim=0.4in 0.6in 0.4in 0.8in]{figures/plot_rot_seq_odd.pdf}
\caption{Two examples of Lemma~\ref{lem:partials-leq-linear-even} with $n=10$.
The blue line shows two periods of the partial sums $B^{(0)}_j$, separated by vertical green lines,
the black line shows $y=x/2$, which was shifted up to the red line to pass through the circled point $(i,B^{(0)}_{i})$ for $i=j^*$ on the left and $i=j'$ on the right (with the orange line showing the odd $j^*$ we couldn't use).}
\label{fig:partials-leq-linear-even}
\end{figure}
}

\begin{corollary}
\label{lem:partials-leq-linear-odd}
	For any such $a$, \Call{rotate-bounded}{$a, i$} holds for some odd index $i$.
\end{corollary}

\begin{proof}
	Consider $a' = (a_1,a_2,\ldots,a_{2n-1},a_0)$. By Lemma~\ref{lem:partials-leq-linear-even}, there is an even index $i'$ which satisfies \Call{rotate-bounded}{$a', i'$}. Thus $i \equiv i'+1$ is an odd index such that \Call{rotate-bounded}{$a, i$}.
\end{proof}


\subsection{Construction Outline}
First, we introduce our reduction from the Team DFA Game to the Team DFA Game with Communication.
Given an $(r,x_0,x_1,N)$-rate-limited policy $P$ and an underlying DFA $D$, we create a DFA $D'$ for playing the TDGC under $P$ which simulates playing the TDG on $D$ while completely clogging the communication between the $\exists$ team to nullify any advantage such communication could bring.
This lets us conclude that a winning strategy for TDGC on $D'$ exists exactly when a winning strategy exists for TDG on $D$.

The reduction applies when each $x_i < r$, meaning the communication rate defined by $P$ is eventually below an average of one bit per round.
We also assume $r>1$ and each $x_i>0$: if there is no communication at all then TDGC is identical to TDG, and if communication only occurs in one direction then the aspects of the construction that deal with the silent direction may be omitted.
Lastly, by Lemma~\ref{lem:policy-period-double}, we take period length $r$ to be even without loss of generality.

The code for $D'$ is fully shown in Algorithm~\ref{alg:tdgc-d-prime-update}. 
The behavior of $D'$ is designed around what we call the \emph{honest} strategy for the $\exists$ team. We will show that it is the only strategy that guarantees the $\exists$ team will pass validation checks by $D'$, but also that it requires using all available transmission bits, resulting in no information transfer between players for their additional benefit in the simulated TDG.

Along with the current state of $D$ in the TDG, $D'$ maintains two queues $X_0,X_1$ of clogging bits that have been given to each $\exists$ player by the adversary $\forall$ in specific rounds. These bits are expected to be submitted by the opposite $\exists$ player to $D'$ for validation in later rounds in order to avoid losing the game, so the players are forced to use transmission bandwidth to exchange this information.
The honest $\exists_i$ player sends these bits directly and as soon as possible to $\exists_{1-i}$, who maintains a ``knowledge'' queue $K_{1-i}$ of all bits sent from $\exists_i$ but not yet validated by $D'$. We note that $X_i \setminus K_{1-i}$ is thus the set of yet-to-be-transmitted private bits known only to $\exists_i$.

\subsection{Build-up Phase}

$D'$ begins the build-up phase after $N$ rounds, once $P$ has started to repeat its policy states.
This phase lasts for exactly $r^2$ rounds, or $r$ periods of $P$'s cycle. 
$D'$ starts with empty $X_0$ and $X_1$, and every round $D'$ simply enqueues $b_0$ and $b_1$ into the appropriate queues. 

During these rounds, $\forall$ can send one bit per round to $\exists_i$, who can transmit those bits to $\exists_{1-i}$, for each $i \in \{0,1\}$.
Because the rate of transmission can vary above or below one bit per round, there is some maximum amount $x'_i \leq x_i$ out of $r$ bits that can reach $\exists_{1-i}$ in the first $r$ rounds.
Each subsequent $r$ rounds, $x_i$ out of the $r$ new bits can be sent (by the rate-limitedness of $P$), thus after $r^2$ rounds at most
$x_i' + (r-1) x_i
\leq r x_i
< r(r-1)$ bits in $X_i$ can be sent to $\exists_{1-i}$ and thus at least $r$ are not known.
By this argument, at the end of the build-up phase we can say that an honest player's knowledge queue has size
$|K_i|
= x_{1-i}' + (r-1) x_{1-i}
\in [r-1,r(r-1)]$,
since we assume $x_{1-i} \geq 1$.

\subsection{Clogging Phase}

In the clogging phase, $D'$ simulates playing TDG on $D$ while clogging the transmissions between $\exists$~players at a steady rate to keep $|X_i|$ and $|K_i|$ constant on period boundaries.
In the last round of every period of $r$ rounds, $D'$ alternates between (1) having $\forall,\exists_0,\exists_1$ play one round of TDG, and (2) forcing $\forall$ to tell the $\exists$~team if they have won in the TDG yet and therefore if $D'$ is going to start the next phase: the tear-down phase.

In the first $r-1$ rounds of each period in this phase, $D'$ clogs the transmissions between $\exists$~players by requiring that bits given to $\exists_i$ by $\forall$ (placed into queue $X_i$) are sent to $\exists_{1-i}$.
This is done by dequeuing the oldest bit $b$ from $X_i$ and checking for $\exists_{1-i}$ to submit $m_{1-i} = b$, otherwise they will lose the game.
Specifically, to preserve the size of $X_i$ and keep up with the rates at which the $\exists$~players can transmit information to each other, $D'$ will do $\Call{enq}{X_i,b_i}$ then validate $\Call{deq}{X_i}=m_{1-i}$ for the first $x_i$ rounds of each period.

Across the whole period, $K_i$ will gain $x_{1-i}$ new bits transmitted from $\exists_{1-i}$ (by the rate-limitedness of $P$). New available bits always exist because the number of private bits available to be sent is $|X_{1-i} \setminus K_i| \geq r > x_{1-i}$ at the start of the period.
Additionally, across the first $x_{1-i}$ rounds of the period, $K_i$ will lose $x_{1-i}$ bits submitted by $\exists_i$ to $D'$, which are always known because $|K_i| \geq (r-1)x_{1-i} \geq x_{1-i}$ at the start.
Overall, this means $|K_i|$ is preserved on period boundaries and honest players will always be able to submit the correct bit and pass the validation.

Labeling the first clogging period with index $0$, at the end of every odd-indexed period, $D'$ will simulate TDG by forwarding the inputs of all players directly to $D$.
However, at the end of even-indexed periods, $D'$ will ignore $\exists_0,\exists_1$ and expect both $\forall$ bits to state whether or not the $\exists$ team has won in TDG, specifically requiring that $b_0=b_1=[q \in F_{\exists}]$.
If this validation fails, then $D'$ will halt with a $\exists$~team victory, so the $\forall$~player must give the correct information to both $\exists$~players to avoid losing, which it is always able to do.

Assuming validation never fails, which is achieved by the honest strategy, the clogging phase continues until the simulated Team DFA Game ends.
If the $\exists$~players lose in the simulation, they lose immediately, otherwise after the even-indexed period when the $\exists$~players learn they have won, $D'$ moves onto the tear-down phase to perform the final validation checks.\looseness=-1

\subsection{Tear-down Phase}

The tear-down phase starts at the beginning of a period, so by the previous arguments for queue size preservation, it starts with $|X_i| = r^2$ and $|K_{1-i}| = x_i' + (r-1) x_i$.
In order to ensure the $\exists$ team's transmissions have been completely clogged all the way until the simulated victory, $D'$ must validate that the remaining bits in $K_i$ have actually been sent by this point.

This phase is split into two parts, with a boosting sub-phase to adjust the size of $X_i$ and $K_{1-i}$ for the following draining sub-phase that empties them. Once each queue has been drained and all validation checks have been passed, then $D'$ will halt with an $\exists$ team victory.
We will need the following fact:

\begin{lemma}
\label{lem:one-sided-tear-down-upper-bound}
	There exists a $k_{end}$ such that, in every round up to the $k_{end}^{\text{th}}$ round of a period, the cumulative number of bits $\exists_0$ will transmit to $\exists_1$ before $\exists_1$ submits a bit to $D'$ in the $k_{end}^{\text{th}}$ round is always upper-bounded by the cumulative number of bits $\exists_1$ will submit to $D'$ in that time (from round $N$ onwards).
\end{lemma}
\begin{proof}
	Say the period begins in round $m \geq N$, and recall that we can assume the period length $r$ is even.
	Consider the sequence $a_k$ of the number of bits transmitted from $\exists_0$ to $\exists_1$ in the $k$ half-rounds starting in round $m$, so $a_k = P_\textsc{mid}(p_{m + k/2})[0]$ when $k$ is even and $a_k = P_\textsc{end}(p_{m + (k-1)/2})[0]$ when $k$ is odd.
	Since policy states repeat, $\forall k \geq 0.\ a_{k+r} = a_k$, and $a_0+\ldots+a_{r-1} = x_0 < r$,
	so we can apply Corollary~\ref{lem:partials-leq-linear-odd} to the reversed sequence
    $(a_{r-1},\ldots,a_0)$
    to get an odd index $i$ such that
    $\forall j>0.\ B^{(i)}_j < \frac{j}{2}$.

    Since $B^{(i)}_j$ is the cumulative number of bits transmitted from $\exists_0$ to $\exists_1$ across the $j$ half-rounds ending when $\exists_1$ submits a bit to $D'$ in round $m+\frac{r-1-i}{2}$,
	and $\lceil \frac{j+1}{2} \rceil \geq \frac{j}{2}$ is the cumulative number of bits $\exists_1$ submits to $D'$ across the same set of $j$ half-rounds, then
	round offset $k_{end} = \frac{r-i-1}{2}$ satisfies Lemma~\ref{lem:one-sided-tear-down-upper-bound}.
\end{proof}
    
\paragraph{Draining Sub-Phase}

    Given $k_{end}$ from Lemma~\ref{lem:one-sided-tear-down-upper-bound} (by symmetry, the lemma applies in both directions), let $t_{end} \leq k_{end}$ be the total number of bits transmitted from the beginning of a period until the bit submission in the $k_{end}^{\text{th}}$ round.
    If a period starts with $|X_i| \leq k_{end}$ and $|X_i \setminus K_{1-i}| = t_{end}$, then we can have $D'$ validate bits in the $|X_i|$ rounds before the $k_{end}^{th}$ round and reach $|X_i|=|K_{1-i}|=0$ where each of the $t_{end}$ transmitted bits are clogging bits from $X_i \setminus K_{1-i}$ with no room for extra communication from $\exists_i$ to $\exists_{1-i}$.
    
    In order to ensure some period starts with $|X_i| \leq k_{end}$ and $|X_i \setminus K_{1-i}| = t_{end}$ we use some $n_i$ periods beforehand to drain each queue appropriately. Since in each period there are $x_i$ transmission bits (fixed) and up to $r$ validated bits (based on $D'$),
    it suffices to have $|X_i| \leq n_i r + k_{end}$
    and $|X_i \setminus K_{1-i}| = n_i x_i + t_{end}$.

\paragraph{Boosting Sub-Phase}
    
    The tear-down phase must start with 
    $|X_i \setminus K_{1-i}| = r^2 - (x_i' + (r-1) x_i) \geq r-1$,
    but this may not be $n_i x_i + t_{end}$ for any $n_i$,
    so before $n_i + 1$ draining periods, we will have additional periods to increase the number of private bits by
    $\delta_i = (n_i x_i + t_{end}) - (r^2 - (x_i' + (r-1) x_i))$.
    So for $\delta_i \geq 0$, we can choose any sufficiently-large $n_i$.

    After one period where $\forall$ gives $c_i$ new clogging bits to $\exists_i$ and $D'$ validates $v_{1-i}$ bits from $\exists_{1-i}$,
    we would have $\Delta|X_i| = c_i - v_{1-i}$
    and $\Delta|K_{1-i}| = x_i - v_{1-i}$ (given that $\exists_i$ initially has $|X_i \setminus K_{1-i}| \geq x_i$ private bits to transmit to $\exists_{1-i}$),
    thus $\Delta|X_i \setminus K_{1-i}|
    = c_i - x_i$.
    Therefore, if we set $c_i = x_i + 1 \leq r$ and $v_{1-i} = x_i$,
    then we get $\Delta|X_i| = +1$,
    $\Delta|K_{1-i}| = 0$,
    and $\Delta|X_i \setminus K_{1-i}| = +1$.
    If $\delta_i < \delta_{1-i}$, then to delay we also need ``filler'' rounds with no change to the sizes of any queues, which can be achieved by setting $c_i = v_{1-i} = x_i$.
    
    To ensure $\delta_i$ is positive and $|X_i| \leq n_i r + k_{end}$ at the \emph{end} of this sub-phase, we need to choose an $n_i$ that satisfies the following constraints at the \emph{start} of the tear-down phase:
    \begin{equation*}
    \begin{aligned}
        0 &\leq \delta_i = (n_i x_i + t_{end}) - |X_i \setminus K_{1-i}| \\
	    n_i &\geq \left( |X_i \setminus K_{1-i}| - t_{end}\right)/x_i 
    \end{aligned}
    \end{equation*}
    and
    \begin{equation*}
    \begin{aligned}
	    n_i r + k_{end} &\geq r^2 + \delta_i \\
	    n_i r + k_{end} &\geq |X_i| + (n_i x_i + t_{end}) - |X_i \setminus K_{1-i}| = |K_{1-i}| + (n_i x_i + t_{end}) \\
	    n_i &\geq \left(|K_{1-i}| + t_{end} - k_{end}\right)/\left(r - x_i\right) 
    \end{aligned}
    \end{equation*}
    
    We pick $n_i$ to be the smallest natural number satisfying both lower bounds:
    \begin{equation*}
    \begin{aligned}
	    n_i &= \left\lceil \max\left\{
	        \frac{|X_i \setminus K_{1-i}| - t_{end}}{x_i},
	        \frac{|K_{1-i}| + t_{end} - k_{end}}{r - x_i}
        \right\} \right\rceil\\
        &= \left\lceil \max\left\{
	        \frac{r^2 - (x_i' + (r-1) x_i) - t_{end}}{x_i} ,
	        \frac{(x_i' + (r-1) x_i) + t_{end} - k_{end}}{r - x_i}
        \right\} \right\rceil\\
    \end{aligned}
    \end{equation*}
    
    Since $0 \leq x_i' \leq x_i < r$ and $0 \leq t_{end} \leq k_{end} < r$, we can upper bound $n_i=O(r^2)$.

\paragraph{Putting it all together}
    
    At the beginning of the tear-down period, $D'$ will run a set of $\delta_i$ periods where $\forall$ produces $x_i + 1$ new bits and $D'$ validates $x_i$ bits, followed by $\max\{\delta_0,\delta_1\} - \delta_i$ periods of $x_i$ new and validated bits.
    After $\delta_{\max} = \max\{\delta_0,\delta_1\}$ rounds, we will have
    $|X_i| = r^2 + \delta_i$
    and $|X_i \setminus K_{1-i}| = n_i x_i + t_{end}$,
    preserving $|K_{1-i}| = x_i' + (r-1) x_i$.
    $D'$ will then run $n_i$ periods plus $k_{end}$ rounds ignoring $\forall$ and validating the remainder of $X_i$ (starting $|X_i|$ rounds before the end).

\subsection{Proof of Undecidability}

\undecidable*

\begin{proof}
We reduce from the Team DFA Game.
For any $(r,x_0,x_1,N)$-rate-limited policy $P$ where $x_0,x_1 < r$,
given an input DFA $D$ for playing the TDG,
we construct the DFA $D'$ described in Algorithm~\ref{alg:tdgc-d-prime-update} 
for playing the TDGC under policy $P$.
Since determining whether or not the $\exists$ team has a forced win in the TDG is undecidable, this reduction will show that the same question of the TDGC under policy $P$ is undecidable as well.

Given the analysis of $D'$ from the previous sections, we first note that $D'$ is indeed a finite automaton: the waiting counter takes on $N$ values; each queue $X_i$ has maximum size $r^2 + \delta_i$ bits, where $n_i = O(r^2)$ so $\delta_i = O(r^3)$; the state $q$ of $D$ has $|Q|$ possible values; and the various other counters require $O(\log r)$ bits each.
From beginning to end, the maximum memory requirement is $O\left(\max\left\{\log N, r^2 + \log|Q|, r^3\right\}\right)$ bits, summarizing Table~\ref{tbl:tdgc-d-prime-memory}.

\begin{table}[htp]
\label{tbl:tdgc-d-prime-memory}
\centering
\begin{tabular}{|l|l|}
\hline
	\textbf{State Category} & \textbf{Space Needed (bits)} \\
\hline
	\Call{halt}{$winner$} & $\Theta(1)$ \\
\hline
	\Call{waiting}{$w$} & $\Theta(\log N)$ \\
\hline
	\Call{build-up}{$X_0$, $X_1$} & $\Theta(r^2)$ \\
\hline
	\Call{clog}{$X_0$, $X_1$, $q$, $p$, $k$, $c_{01}$, $c_{10}$} & $\Theta(r^2 + \log|Q|)$ \\
\hline
	\Call{boost}{$X_0$, $X_1$, $d_{01}$, $d_{10}$, $k$, $c_{01}$, $c_{10}$} & $\Theta(r^2 + \delta_{\max})$ \\
\hline
	\Call{drain}{$X_0$, $X_1$, $c_{01}$, $c_{10}$} & $\Theta(r^2 + \delta_{\max})$ \\
\hline
\end{tabular}
\caption{Memory Requirements of $D'$ over the course of the TDGC.}
\end{table}

If there is a winning strategy $S$ for the $\exists$ team on $D$ in the TDG, then the corresponding honest strategy described above that plays the simulated TDG using $S$ will be a winning strategy for the $\exists$ team on $D'$ in the TDGC under policy $P$.

If there is a winning strategy for the $\exists$ team on $D'$ in the TDGC under policy $P$, then consider any winning execution $\gamma$. Since winning requires termination, let $C$ be the number of periods in the clogging phase.

If $\gamma$ reaches \Call{halt}{$\exists$} in the clogging phase because $\forall$ did not correctly tell the $\exists$ team whether or not $q \in F_\exists$, then $\forall$ did not play optimally.
Since $\forall$ has perfect information and is allowed to give either $0$ or $1$ by the game rules, there is an alternate execution $\gamma'$ where $\forall$ gives the correct answer instead and the game continues, so no $\exists$ team strategy can force a win in this way.

The only other way for the $\exists$ team to win is for $\gamma$ to reach \Call{halt}{$\exists$} at the end of the draining phase, which means they must pass all of the validation checks by $D'$.

\begin{table}[htp]
\centering
\begin{tabular}{|l|l|l|l|}
\hline
	Phase & \Call{enq}{$X_i$} Count & \Call{deq}{$X_i$} Count & Information $\exists_i \to \exists_{1-i}$ \\
\hline
	Build-up & $r^2$ & $0$ & $x_i' + (r-1) \times x_i$ \\
\hline
	Clogging & $C \times x_i$ & $C \times x_i$ & $C \times x_i$\\
\hline
	Boosting & $\delta_{\max} \times x_i + \delta_i$ & $\delta_{\max} \times x_i$ & $\delta_{\max} \times x_i$\\
\hline
    Draining & $0$ & $r^2 + \delta_i$ & $n_i \times x_i + t_{end_i}$ \\
\hline
\end{tabular}
\caption{Accounting of \Call{enq}{$X_i$}, \Call{deq}{$X_i$}, and Information Transfer between players in each phase}
\label{tbl:d-prime-accounting}
\end{table}

Table~\ref{tbl:d-prime-accounting} details the value of three quantities in each phase of the game: the number of bits enqueued into $X_i$, the number of bits dequeued from $X_i$, and the amount of meaningful bits of information that can be transmitted from $\exists_i$ to $\exists_{1-i}$.
By the definition of $\delta_i$ and some algebra, it can be seen that each column has the same sum;
let $I$ be this total quantity of bits.

Because $D'$ validates the value of each dequeued bit, in order for $\exists_i$ to guarantee they pass all validation checks, they must send $I$ bits of information to $\exists_{1-i}$.
However, because $I$ is the maximum amount of information $\exists_i$ can send to $\exists_{1-i}$, no further information can be sent,
which means that in every round in which $D'$ simulates the TDG on $D'$, $\exists_i$ has the same amount of information about the state $q$ of $D$ as it would when actually playing TDG on $D$.
Therefore, if the $\exists$ team has a winning strategy for playing TDGC on $D'$ under policy $P$, within it is a winning sub-strategy for them to play the TDG on $D$.
\end{proof}

\begin{algorithm}
\caption{Pseudocode for the $D'$ internal update function per round}
\label{alg:tdgc-d-prime-update}
\begin{algorithmic}[1]
\small
\State $q' \gets \Call{waiting}{1}$ \Comment{Initial state}
\Function{dfa-round-update}{$q'$, $b_0$, $b_1$, $m_0$, $m_1$}
    \Switch{$q'$}
        \Case{\Call{halt}{$winner$}} \Comment{Game is over, with $q' \in F'_{winner}$}
            \State \Return \Call{halt}{$winner$}
        \EndCase
        \Case{\Call{waiting}{$w$}} \Comment{Waiting Phase, delaying until policy starts repeating}
            \IfOneLine{$w<N$}{\Return \Call{waiting}{$w+1$}}
            \State \Return \Call{build-up}{$[]$, $[]$}
        \EndCase
        \Case{\Call{build-up}{$X_0$, $X_1$}} \Comment{Build-up Phase, filling up $X_i$ queues}
            \State \Call{enq}{$X_0$, $b_0$}
            \State \Call{enq}{$X_1$, $b_1$}
	        \IfOneLine{$\Call{length}{X_0} < r^2$}{\Return \Call{build-up}{$X_0$, $X_1$}}
	        \State \Return \Call{clog}{$X_0$, $X_1$, $q_0$, $\Call{even}{}$, $r$, $x_0$, $x_1$}
        \EndCase
        \Case{\Call{clog}{$X_0$, $X_1$, $q$, $p$, $k$, $c_{01}$, $c_{10}$} \textbf{given} $k > 1$} \Comment{Clogging Phase, boosting}
            \ForAll{$i \in \{0,1\}$}
                \If{$c_{i,1-i} > 0$}
                    \State \Call{enq}{$X_i$, $b_i$}
                    \IfOneLine{$\Call{deq}{X_i} \neq m_{1-i}$}{\Return $\Call{halt}{\forall}$}
                    \State $c_{i,1-i} \gets c_{i,1-i} - 1$
                \EndIf
            \EndFor
            \State \Return \Call{clog}{$X_0$, $X_1$, $q$, $p$, $k-1$, $c_{01}$, $c_{10}$}
        \EndCase
        \Case{\Call{clog}{$X_0$, $X_1$, $q$, $\Call{odd}{}$, $1$, $0$, $0$}} \Comment{Clogging Phase, simulating $D$}
            \State $q \gets \delta(\delta(\delta(\delta(q, b_0), b_1), m_0), m_1)$
            \State \Return \Call{clog}{$X_0$, $X_1$, $q$, $\Call{even}{}$, $r$, $x_0$, $x_1$}
        \EndCase
        \Case{\Call{clog}{$X_0$, $X_1$, $q$, $\Call{even}{}$, $1$, $0$, $0$}} \Comment{Clogging Phase, testing for $\exists$ win}
            \IfOneLine{$\lnot \left(b_0 = b_1 = [q \in F_{\exists}]\right)$}{\Return $\Call{halt}{\exists}$}
            \IfOneLine{$q \in F_\exists$}{\Return \Call{boost}{$X_0$, $X_1$, $\delta_0$, $\delta_1$, $r$, $x_0$, $x_1$}}
            \IfOneLine{$q \in F_\forall$}{\Return $\Call{halt}{\forall}$}
            \State \Return \Call{clog}{$X_0$, $X_1$, $q$, $\Call{odd}{}$, $r$, $x_0$, $x_1$}
        \EndCase
        \Case{\Call{boost}{$X_0$, $X_1$, $d_{01}$, $d_{10}$, $k$, $c_{01}$, $c_{10}$} \textbf{given} $k>1$} \Comment{Boost Phase, clogging}
            \ForAll{$i \in \{0,1\}$}
                \If{$c_{i,1-i} > 0$}
                    \State \Call{enq}{$X_i$, $b_i$}
                    \IfOneLine{$\Call{deq}{X_i} \neq m_{1-i}$}{\Return $\Call{halt}{\forall}$} \Comment{Return from caller}
                    \State $c_{i,1-i} \gets c_{i,1-i} - 1$
                \EndIf
            \EndFor
            \State \Return \Call{boost}{$X_0$, $X_1$, $d_{01}$, $d_{10}$, $k-1$, $c_{01}$, $c_{10}$}
        \EndCase
        \Case{\Call{boost}{$X_0$, $X_1$, $d_{01}$, $d_{10}$, $1$, $0$, $0$}} \Comment{Boost Phase, new boost bits}
            \ForAll{$i \in \{0,1\}$}
                \If{$d_{i,1-i} > 0$}
                    \State \Call{enq}{$X_i$, $b_i$}
                    \State $d_{i,1-i} \gets d_{i,1-i} - 1$
                \EndIf
            \EndFor
            \IfOneLine{$d_{01} + d_{10} > 0$}{\Return \Call{boost}{$X_0$, $X_1$, $d_{01}$, $d_{10}$, $r$, $x_0$, $x_1$}}
            \State \Return \Call{drain}{$X_0$, $X_1$, $r \times n_0 + k_{end_0}$, $r \times n_1 + k_{end_1}$}
        \EndCase
        \Case{\Call{drain}{$X_0$, $X_1$, $c_{01}$, $c_{10}$} given $c_{01} + c_{10} > 0$} \Comment{Drain Phase, emptying queues}
            \ForAll{$i \in \{0,1\}$}
                \If{$c_{i,1-i1} > 0$}
                    \IfOneLine{$|X_i| = c_{i,1-i} \land \Call{deq}{X_i} \neq m_{1-i}$}{\Return $\Call{halt}{\forall}$}
                    \State $c_{i,1-i1} \gets c_{i,1-i1} - 1$
                \EndIf
            \EndFor
            \State \Return \Call{drain}{$X_0$, $X_1$, $c_{01}$, $c_{10}$}
        \EndCase
        \Case{\Call{drain}{$[]$, $[]$, $0$, $0$}} \Comment{Drain Phase, finished!}
            \State \Return \Call{halt}{$\exists$}
        \EndCase
    \EndSwitch
\EndFunction

\end{algorithmic}
\end{algorithm}

\section{Decidability}
\label{sec: Decidability}

We show that our general construction from the previous section is tight with respect to the transmission rate between $\exists$ players. 

For our precise bounds, we assume the straightforward encoding of the input DFA $D$ with $n$ states as a table for $\delta$ containing $2n$ states, a state $q_0$, and the states in $F_{\exists}$ and $F_{\forall}$, thus the input size is $\Theta(|Q|)$.

First, we demonstrate $(r,r,r,0)$-rate-limited policies under which the Team DFA Game with Communication is not only decidable but in PSPACE. Later we will show more restrictive communication patterns are in EXPSPACE. Recall $(r,r,r,0)$-rate-limited policies are the case where both players are allowed to exchange $r$ bits over the course of a period of length $r$.

\begin{theorem}
\label{thm:tdgc-1-bit-every-mid-round-decidable}
TDGC is decidable in PSPACE with a 1-bit, mid-round exchange in both directions every round: policies $P$ with
$P_\textsc{mid}(p) = (1,1)$
and
$P_\textsc{end}(p) = (0,0)$ for all $p \in \Pi$.
\end{theorem}
\begin{proof}
    Under such a policy, TDGC becomes a perfect information game.
    In each round of the game, the optimal play for $\exists_i$ is to send $b_i$ to $\exists_{1-i}$ immediately after receiving it, meaning $\exists_{1-i}$ will know both $b_0$ and $b_1$ before it chooses $m_{1-i}$.
    Since the $\exists$~team knows the initial state $q_0$ of $D$,
    we can consider strategy functions $s : (q,b_0,b_1) \mapsto (m_0,m_1)$, which both players can use to decide their own next move and know what move their teammate will perform as well, letting them use $\delta$ to learn the state $q$ of $D$ in the next round and beyond.
    
    Note that it suffices for the $\exists$~team to have a memoryless strategy because the policy $P$ is constant per round, DFA transitions do not depend on the history of the game, and the adversarial $\forall$ player's choices are not bound by the history either.
    It also suffices to have a deterministic strategy: if there exists a non-deterministic winning strategy $s'$, then we can fix $s(q,b_0,b_0)$ to be some $(m_0,m_1)$ with $\Pr[s'(q,b_0,b_1) = (m_0,m_1)] > 0$ because all game executions in which the $\exists$~team plays with deterministic strategy $s$ are possible executions when playing with strategy $s'$, thus must also be winning.
    
    We show that deciding whether or not the $\exists$~team has a forced win in TDGC under policy $P$ is in PSPACE by giving a brute-force search algorithm.
    For every strategy $s$ among the $4^{4|Q|}$ possible strategy functions, 
    \iftrue
    we construct a game graph $G_s$ where each state $q \in Q \setminus F_{\exists}$ is a vertex and for all $b_0,b_1 \in \{0,1\}$, $q$ has an edge to $q' = \delta(\delta(\delta(\delta(q, b_0), b_1), m_0), m_1)$ where $(m_0,m_1) = s(q,b_0,b_1)$ as long as $q' \notin F_{\exists}$.
    This means $s$ is a winning strategy if and only if all $q \in F_{\forall}$ and all cycles are not reachable from $q_0$ in $G_s$, since otherwise the traversal corresponds to a losing execution or the start of a potentially non-terminating execution of the game that the $\forall$~player can force to occur.
    We can thus perform an exhaustive depth-first search from $q_0$ for a counterexample (of length at most $|Q|$) to decide whether or not $s$ is a winning strategy.
    Since we only need $\Theta(|Q|)$ space to store the current $s$, $G_s$, and depth-first search stack, this algorithm runs in PSPACE.
    \else
    we can test $s$ against all counter-strategies for the $\forall$~player -- consisting of the bits $(b_0(q),b_1(q))$ that $\forall$ should play when in state $q \in Q$ -- for one which does not lead to a win when the $\exists$~team plays according to $s$.
    The test is a simulation of the TDGC under $P$, where in DFA state $q$ the $\forall$~player submits $(b_0(q),b_1(q))$ then the $\exists$~team submits $s(q,b_0(q),b_1(q)) = (m_0(q),m_1(q))$.
    This simulation only needs to run for up to $|Q|$ rounds because if the game repeats a state before either team wins, then the game will not terminate and thus does not lead to a win for the $\exists$~team.
    Since we only need $\Theta(|Q|)$ space to store $s$, $b_0(q)$, $b_1(q)$, and the simulation's game state $q \in |Q|$, this algorithm is in PSPACE.
    \fi
\end{proof}

Since it is sufficient to send only one bit of useful information mid-round, we can extend Theorem~\ref{thm:tdgc-1-bit-every-mid-round-decidable} to higher transmission rates.

\begin{corollary}
	TDGC is decidable in PSPACE with at least a 1-bit, mid-round exchange in both directions every round: policies $P$ with
$P_\textsc{mid}(p)[i] \geq 1$ for all $p \in \Pi$ and each $i \in \{0,1\}$.
\end{corollary}


Next, we consider the decidability of TDGC under $(r,r,0,0)$-rate-limited policies, which is tight given the undecidability of $(r,r-1,0,0)$-rate-limited policies. This shows that only one member of the team needs to have perfect information.

\begin{theorem}
\label{thm:tdgc-1-bit-every-mid-round-one-way-decidable}
TDGC is decidable in EXPSPACE with a 1-bit, mid-round exchange every round from $\exists_0$ to $\exists_1$, but none from $\exists_1$ to $\exists_0$: policies $P$ with
$P_\textsc{mid}(p) = (1,0)$
and
$P_\textsc{end}(p) = (0,0)$ for all $p \in \Pi$.
\end{theorem}
\begin{proof}
    As described in the proof of Theorem~\ref{thm:tdgc-1-bit-every-mid-round-decidable}, $\exists_0$ can and should send $b_0$ to $\exists_1$ each round to give $\exists_1$ perfect information, but $\exists_0$ themself can learn nothing about $b_1$.
    Using the terminology from \cite{peterson_multiple-person_1979}, this asymmetry makes TDGC under $P$ a hierarchical team game.
    To decide the existence of a winning strategy, we adapt ideas from the proof of Theorem~4 in the same paper that shows DTIME$\left(2^{2^{2^{cS(n)}}}\right) \supseteq$ MPA${}_2$-SPACE$(S(n))$, the languages decided by hierarchical 2-vs-1 private alternation Turing machines in $S(n)$ space.
    
    Consider the set of all possible mid-round configurations $(q,b_0,b_1)$ of the game, which are fully known to $\forall$ and $\exists_1$.
    Define $C$ be the set of possible configurations $(b_0,u)$ of $\exists_0$'s mid-round knowledge: the known $b_0$ and the set $u \in \mathcal{P}(Q \times \{b_0\} \times \{0,1\})$ of possible mid-round configurations given the history of the game thus far.
    Since two game states with the same $c \in C$ are strategically equivalent from the perspective of $\exists_0$ (and thus $\exists_1$ too), a winning strategy only needs to account for the $|C| = 2^{2|Q|+1}$ knowledge configurations in its decision-making.\looseness=-1
    
    Given this, we can do a brute-force search as in Theorem~\ref{thm:tdgc-1-bit-every-mid-round-decidable} over the space of deterministic $\exists$~team strategies $s : c \in C \mapsto (m_0,m_1)$ of size $4^{|C|}$.
    For each $s$, we construct the game graph $G_s$, where $c \in C$ has an outgoing edge representing the outcome of each $b_0,b_1$ choice of $\forall$ after the $\exists$~players use $s$ to make their moves and $\exists_0$ updates their knowledge, and then search for counter-example game executions with length up to $|C|$ to decide whether $s$ is a winning strategy.
    Therefore, TDGC under $P$ is decidable in $\Theta(|C|)$ space, which is exponential in $|Q|$.
\end{proof}

As before, Theorem~\ref{thm:tdgc-1-bit-every-mid-round-decidable} extends to higher transmission rates from $\exists_0$ to $\exists_1$ (or vice versa), as long as the receiver stays silent.

\begin{corollary}
	TDGC is decidable in PSPACE with at least a 1-bit, mid-round exchange in one direction every round, but none in the other direction: policies $P$ with
$P_\textsc{mid}(p)[i] \geq 1$ and $P_\textsc{mid}(p)[1-i] = P_\textsc{end}(p)[1-i] = 0$ for all $p \in \Pi$ and some $i \in \{0,1\}$.
\end{corollary}


\section{Team Formula Games with Communication}
\label{sec: Team Formula Game}

\ifabstract
\later{
\section{Team Formula Games with Communication}
\label{app: Team Formula Game}
}
\fi

Formula games model many types of games.
The Team Formula Game was defined and proven undecidable in \cite{demaine2008constraint}. We define a communication version of this game and prove results analogous to the ones for TDGC. 

\begin{define}
	A \emph{Team Formula Game} (TFG) instance consists of sets of Boolean variables $X$, $X'$, $Y_1$, $Y_2$ and their initial values; variables $h_0,h_1 \in X$; and Boolean formulas $F(X,X',Y_0,Y_1)$, $F'(X,X')$, and $G(X)$ such that $F$ implies $\lnot F'$.
	The TFG problem asks whether $\{W_0,W_1\}$, team White, has a forced win against $\{B\}$, team Black, in the game that repeats the following steps in order ad infinitum:\looseness=-1
	\begin{enumerate}
	    \item \label{enum:tfg-set-X}
	    $B$ sets $X$ to any values. If $F$ and $G$ are true, then Black wins. If $F$ is false, White wins.
	    
	    \item \label{enum:tfg-set-X-prime}
	    $B$ sets $X'$ to any values. If $F'$ is false, then White wins.
	    
	    \item \label{enum:tfg-mid}
	    $W_0$ sets $Y_1$ to any values.
	    
	    \item \label{enum:tfg-end}
	    $W_1$ sets $Y_2$ to any values.
    \end{enumerate}
    where $B$ has perfect information but $W_i$ can only see the values of $Y_i$ and $h_i$.
\end{define}

\begin{define}
	\emph{Team Formula Game with Communication} (TFGC) is TFG along with a policy $P$ which specifies a number of bits to be transmitted between $W_0$ and $W_1$ mid-round (before each step~\ref{enum:tfg-mid}) and at the end of the round (after each step~\ref{enum:tfg-end})\looseness=-1
\end{define}

\begin{theorem}
	TFGC is undecidable under all $(r,x_0,x_1,N)$-rate-limited policies where $x_0,x_1 < r$.
\end{theorem}
\begin{proof}
	For any such policy $P$, we reduce from the Team DFA Game with Communication under the same policy $P$, adapting the reduction done in Theorem~8 of \cite{demaine2008constraint} from the Team Computation Game to the Team Formula Game.
	In the reduction, the White team plays as the $\exists$~team and $B$ plays as $\forall$ while also facilitating the simulation of TDGC in TFGC.
	
	Given a DFA $D$ to play TDGC under $P$, we first augment $D$ so it will be suitable for the simulation. To each state, we add a 3-value counter to eliminate any four-edge cycles in the transition graph ($t \xrightarrow{\delta} (t+1) \xrightarrow{\delta} (t+2) \xrightarrow{\delta} t \xrightarrow{\delta} (t+1) \neq t$). Also, we add four new states in a path $q_0 \xrightarrow{\delta} q_0^{(1)} \xrightarrow{\delta} q_0^{(2)} \xrightarrow{\delta} q_0^{(3)} \xrightarrow{\delta} q_0^{(4)}$ from a new initial state $q_0$ to the original initial state $q_0^{(4)}$ in order to delay the first meaningful state transitions until the start of the second round, which is when the first set of player inputs are available.
	
	In the instance of TFGC, we will have
	(1) variables $h_i = b_i \in X$ and $b_i' \in X'$, representing the $\forall$~player's chosen bits in the current and previous round;
	(2) $Y_i = \{m_i\}$, containing the $\exists_i$ player's message bit each round;
	(3) sets of $\Theta(\log|Q|)$ variables $\enc{q'} \subset X'$ and $\enc{q} \subset X$ that encode the previous state $q'$ and current state $q$;
	(4) and two parity bits $p \in X$ and $p' \in X'$ which $B$ will be required to flip each round.
	We also choose the initial value of $q'$ to be $q_0$ so that in step~\ref{enum:tfg-set-X} of the first round $B$ will be forced to set $q$ to $q_0^{(4)}$; other initial values are arbitrary.
	
	In step~\ref{enum:tfg-set-X}, formula~$F$ holds if $B$ sets $X$ so $q = \delta(\delta(\delta(\delta(q', b_0'), b_1'), m_0), m_1)$, $q \notin F_{\exists}$, and $p' \neq p$.
	Formula~$G$ will be true if the current state $q \in F_{\forall}$. Thus, when $F$ and $G$ are both true, then in the TDGC the state transition function was correctly implemented and led to a final state where $\forall$ has won, and thus Black wins the TFGC.
	On the other hand, if $F$ is false, then either $B$ violated the simulation or the TDGC led to a final state where $\exists$~team has won, and thus White wins the TFGC.
	
	In step~\ref{enum:tfg-set-X-prime}, formula~$F'$ will be true if $B$ sets $X'$ such that $q' = q$ and $p' = p$, updating the previous state for the next round to the new state. If $F'$ is false, then $B$ violated the simulation, and thus White wins the TFGC.
	Additionally, the parity bit checks guarantee that $F$ implies $\lnot F'$.
	
	Since this is a faithful simulation where each round of TFGC corresponds exactly to one round of TDGC,
	and by Theorem~\ref{thm:tdgc-undecidable-xi-lt-r} it is undecidable whether or not there exists a winning strategy for the $\exists$~team playing TDGC under $P$,
	it is also undecidable whether or not there exists a winning strategy for White playing TFGC under the same policy $P$.
\end{proof}


The strategy for proving decidability results of Team DFA Game with Communication also be used to give the following tight decidability results on the Team Formula Game with Communication. \ifabstract See Appendix~\ref{app: Team Formula Game} for details.\fi

\both{
\begin{theorem}
\label{thm:tfgc-1-bit-every-mid-round-decidable}
TFGC is decidable in 2-EXPSPACE with a 1-bit, mid-round exchange in both directions every round: policies $P$ with
$P_\textsc{mid}(p) = (1,1)$
and
$P_\textsc{end}(p) = (0,0)$ for all $p \in \Pi$.
\end{theorem}
}

\later{
\begin{proof}
	Since player $B$ gives one bit $h_i$ to player $W_i$ each round, under such a policy $P$, $h_i$ can be immediately transmitted to $W_{1-i}$. Thus $W_0$ and $W_1$ always have the same information when choosing to set $Y_0$ and $Y_1$.
	This effectively makes the TFGC a two-player game, but not a perfect information game: the values of $X' \cup X \setminus \{h_0,h_1\}$ are not visible to White, and between each move of White, Black is allowed to set $X$ and $X'$ to any values.
	However, White does gain information on their turns from knowing $h_0,h_1$ and that the game did not end, as it must have been that $F$ was true and $G$ was false after step~\ref{enum:tfg-set-X}, and $F'$ was true after step~\ref{enum:tfg-set-X-prime}.
	
	Therefore, as before, we can check over the entire space of strategies $s : (c,h_0,h_1) \mapsto (Y_0,Y_1)$ for White, where $c \in C$ specifies the set of possible hidden values for $X' \cup X \setminus \{h_0,h_1\}$, by creating a strategy graph $G_s$ on $C$ with degree $4$ to search for counterexample paths of length at most $|C|$ to $s$ being a winning strategy.
	Since $|C| = 2^{2^{|X|+|X'|-2}}$, we only require space doubly-exponential in the input size, which is $\Omega(|X|+|X'|)$.
\end{proof}
}

\both{
\begin{theorem}
\label{thm:tfgc-1-bit-every-mid-round-one-way-decidable}
TFGC is decidable in 3-EXPSPACE with a 1-bit, mid-round exchange every round from $W_0$ to $W_1$, but none from $W_1$ to $W_0$: policies $P$ with
$P_\textsc{mid}(p) = (1,0)$
and
$P_\textsc{end}(p) = (0,0)$ for all $p \in \Pi$.
\end{theorem}
}

\later{
\begin{proof}
	By the same argument as in Theorem~\ref{thm:tdgc-1-bit-every-mid-round-one-way-decidable}, we can extend Theorem~\ref{thm:tfgc-1-bit-every-mid-round-decidable} by adding the set of possible knowledge states of $W_0$ to the game configuration.
	If we take $C = \mathcal{P}\left(X' \cup X \setminus \{h_0,h_1\}\right)$ then we must consider triplets $(h_0,c_0,c)$ for actual configurations $c \in C$ the game can be in, and $c_0 \in \mathcal{P}(C)$ specifying what $W_0$ thinks $c$ could be.
	
	This allows us to bound the space required to search (like before) for a winning strategy $s : (h_0 \in \{0,1\}, c_0 \in \mathcal{P}(C)) \mapsto (Y \cup Y')$ to $2^{2^{2^{O(|X|+|X'|)}}}$, triply-exponential in the input size, which is $\Omega(|X|+|X'|)$.
\end{proof}
}


\section{Open Problems}

One exciting question is whether we can prove computational complexity results about real games with communication. It seems plausible that TDGC may be sufficient for applications to games with highly structured communication. We present a number of questions that we think may help strengthen results to allow their application to more real world scenarios or questions we find particularly interesting for their own sake.

One of the main technical questions left open by this work is the complexity for rate-limited policies with $x_0 \geq r$ and $r>x_1>0$. We conjecture this case is decidable but our current arguments rely on both players either having full information or no information. 

Looking further, there are many interesting variations and extensions of this model to study. Our arguments rely heavily on communication policies having some bounded period which is useful both for algorithms to bound the uncertainty in the game and for undecidability to allow for constructions that simulate a step in a zero information game after a bounded number of rounds. What happens if our policy is described by something more general than a DFA, such as a sequence recognizable by a pushdown automaton?

Similarly, some of our arguments rely on the fact that the game is played on something with bounded state, such as a DFA or Boolean Formula. What happens with team games on more general systems, such as a pushdown automaton or a bounded space Turing Machine?

Many realistic scenarios have noisy communication channels. How does the computational complexity change under different models of noise? We conjecture that there will again be a cutoff based on whether the information capacity of the channel is sufficiently high. However, it is also possible that the small probability of error will compound over these games of unbounded length resulting in different behavior. It would also be interesting to understand what happens when other sources of inherent randomness are introduced to these games.

It is also often the case that one's ability to communicate depends on the state of the environment and potentially the actions of the people involved. Thus having communication policies that depend on player actions or the game state would be another interesting generalization.

We also only consider two players on the Existential Team. We believe that when more players are added, undecidability will emerge if at least two players have imperfect information. However, this should be verified and the details around more complex communication patterns may lead to richer behavior.

Finally, there is an issue when trying to apply these results to real games or real world problems. Our characterization in some sense relies on communication being high or low compared to critical or meaningful choices in the games. Many natural scenarios have a much larger action space than communication rate, however many of those choices may be essentially equivalent or strategically inadvisable. Undecidability proofs such as those for Team Fortress 2~\cite{coulombe_et_al:LIPIcs:2018:8805} have very inefficient reductions and require significant numbers of in-game actions to simulate one move in the DFA game. This makes a direct application of our results difficult.

\subsubsection*{Acknowledgements}
We would like to thank Erik Demaine and other participants in the class 6.892 Algorithmic Lower Bounds: Fun with Hardness Proofs (Spring 2019) for useful discussion and the suggestion of potential applications. Thanks to Sophie Monahan for editing assistance. 


\bibliographystyle{eptcs}
\bibliography{biblio.bib}{}


\appendix
\magicappendix

\end{document}

%% file: mathalg.tex
\usepackage{amsmath}
\usepackage{amsthm}
\usepackage{amsfonts}
\usepackage{amssymb}
\usepackage[noEnd]{algpseudocodex}
\usepackage{algorithm}
\usepackage{color}

\usepackage{fixltx2e}
\MakeRobust{\Call}

\ifcsname theorem\endcsname
\else
\theoremstyle{plain}
    \newtheorem{theorem}{Theorem}
    \newtheorem{lemma}{Lemma}
    \newtheorem{corollary}{Corollary}
    
\theoremstyle{definition}
    \newtheorem{define}{Definition}
    
\fi

\def\enc#1{\langle #1 \rangle}

\iffalse
    
\else
    
\fi

\iffalse
    
\else
    
\fi

\iftrue

\algblockdefx{TFunction}{EndTFunction}[3]{#1 \textbf{function} \textsc{#2}(#3)}{\textbf{end function}}
\else

\algblockdefx{TFunction}{EndTFunction}[3]{\textbf{function} \textsc{#2}(#3): #1}{\textbf{end function}}
\fi

\algblockdefx{Sync}{EndSync}[1]{\textbf{synchronized} #1 \textbf{do}}{\textbf{end synchronized}}
\algblockdefx{TrySync}{EndTrySync}[1]{\textbf{try synchronized} #1 \textbf{do}}{\textbf{end try synchronized}}
\algdef{CE}[TrySyncElse]{TrySync}{ElseTrySync}{EndTrySync}%
   [1]{\textbf{else} #1}%
   [0]{\textbf{end try synchronized}}

\algblockdefx{ParFor}{EndParFor}[1]{\textbf{for} #1 \textbf{in parallel do}}{\textbf{end for}}

\def\IfOneLine#1#2{\State \textbf{if} #1 \textbf{then} #2}

\iffalse
\algblockdefx{Switch}{EndSwitch}[1]{\textbf{switch} #1 \textbf{on}}{\textbf{end switch}}
\algblockdefx{Case}{EndCase}[1]{\textbf{case} #1 \textbf{do}}{\textbf{end case}}
\algnotext[Case]{EndCase}
\algdef{CE}[Default]{Switch}{Default}{EndSwitch}%
   [0]{\textbf{default}}%
   [0]{\textbf{end switch}}
\else
    \algnewcommand\algorithmicswitch{\textbf{switch}}%
    \algnewcommand\algorithmiccase{\textbf{case}}%

    \makeatletter
    \algdef{SE}[SWITCH]{Switch}{EndSwitch}[1]{\algpx@startIndent\algpx@startCodeCommand\algorithmicswitch\ #1\ \algorithmicdo}{\algpx@endIndent\algpx@startCodeCommand\algorithmicend\ \algorithmicswitch}%
    \algdef{SE}[CASE]{Case}{EndCase}[1]{\algpx@startIndent\algpx@startCodeCommand\algorithmiccase\ #1}{\algpx@endIndent\algpx@startCodeCommand\algorithmicend\ \algorithmiccase}%

    \ifbool{algpx@noEnd}{%
      \algtext*{EndSwitch}%
      \algtext*{EndCase}%
      \apptocmd{\EndSwitch}{\algpx@endIndent}{}{}%
      \apptocmd{\EndCase}{\algpx@endIndent}{}{}%
    }{}%

    \pretocmd{\Switch}{\algpx@endCodeCommand}{}{}
    \pretocmd{\Case}{\algpx@endCodeCommand}{}{}

    \ifbool{algpx@noEnd}{%
      \pretocmd{\EndSwitch}{\algpx@endCodeCommand[1]}{}{}%
      \pretocmd{\EndCase}{\algpx@endCodeCommand[1]}{}{}%
    }{%
      \pretocmd{\EndSwitch}{\algpx@endCodeCommand[0]}{}{}%
      \pretocmd{\EndCase}{\algpx@endCodeCommand[0]}{}{}%
    }%
    \makeatother
\fi

\algloopdefx{Throw}[1]{\textbf{throw} #1}
\algdef{SE}[Try]{Try}{EndTry}%
   [0]{\textbf{try}}{\textbf{end try}}%
\algdef{CE}[Catch]{Try}{Catch}{EndTry}%
   [1]{\textbf{catch} #1}%
   [0]{\textbf{end try}}
\algdef{CE}[Catch]{Catch}{Catch}{EndTry}%
   [1]{\textbf{catch} #1}%
   [0]{\textbf{end try}}

\algblockdefx{Struct}{EndStruct}[1]{\textbf{struct} #1}{\textbf{end struct}}

\algrenewcommand\algorithmicdo{}